\documentclass[conference,twocolumn]{IEEEtran}

\usepackage{amsmath,amssymb,amsfonts}
\usepackage{algorithm}
\usepackage{algpseudocode}
\usepackage{graphicx}
\usepackage{textcomp}
\usepackage{xcolor}
\usepackage[hidelinks]{hyperref}
\usepackage{cite}
\usepackage{booktabs}
\usepackage{placeins}
\usepackage{float}
\usepackage{tikz}
\usetikzlibrary{positioning}
\usepackage{pgfplots}
\pgfplotsset{compat=1.18}
\usepackage{amsthm}
\newtheorem{theorem}{Theorem}
\newtheorem{proposition}{Proposition}

\newcommand{\posme}{\textsc{PoSME}}

\begin{document}

\title{PoSME: Proof of Sequential Memory Execution\\via Latency-Bound Pointer Chasing\\with Causal Hash Binding}

\author{
\IEEEauthorblockN{David Condrey\\\small ORCID: 0009-0003-1849-2963}
\IEEEauthorblockA{WritersLogic Inc.\\
San Diego, California\\
david@writerslogic.com}
}

\maketitle

\begin{abstract}
We introduce \posme{} (Proof of Sequential Memory Execution), a cryptographic primitive that enforces sustained sequential computation via latency-bound pointer chasing over a mutable arena. Each step reads data-dependent addresses, writes a block whose value and causal hash are mutually dependent (symbiotic binding), and chains the result into a global transcript. This yields three properties: (1) $\Omega(K)$ sequential memory-step enforcement, (2) high TMTO resistance ($10{\times}$ at write density $\rho{=}4$, with a formal $S \cdot T = \Omega(K^2)$ space-time lower bound), and (3) a tight ASIC advantage bound by DRAM random-access latency rather than bandwidth. Benchmarks across 17 CPU platforms and 4 GPU architectures demonstrate that hash computation is under 3.5\% of step cost and GPU hardware is \textbf{14--19$\times$ slower} than a consumer CPU. \posme{} requires no trusted setup and provides a foundation for verifiable delay, authorship attestation, and Sybil resistance.
\end{abstract}

\section{Introduction}

Many systems require proof that a party performed sustained,
sequential computation under concrete resource
constraints~\cite{dwork1993,back2002,nakamoto2008}:
ASIC-resistant mining, authorship attestation, Sybil
resistance, and verifiable delay. Existing
primitives each capture only part of this objective.

VDFs~\cite{boneh2018vdf} certify sequential time but are
insensitive to memory; a VDF ASIC gains arbitrary speedup via
faster ALUs. PoSW~\cite{cohen2018posw} proves graph traversal
but over static, immutable memory. Memory-hard functions such as
Argon2id~\cite{biryukov2016} raise per-evaluation cost via memory
pressure, but are single-shot primitives with no chain proof
system. Composing them (e.g., chaining Argon2id with Merkle
sampling) yields independent properties that do not reinforce
each other. Table~\ref{tab:compare} summarizes the gap.

\begin{table}[!ht]
\centering
\caption{Property comparison of sequential primitives.}
\begin{tabular*}{\columnwidth}{@{\extracolsep{\fill}}lcccc@{}}
\toprule
 & VDF & PoSW & MHF & \posme{} \\
\midrule
Sequential steps     & \checkmark & \checkmark & $\times$ & \checkmark \\
Memory-hard         & $\times$ & $\times$ & \checkmark & \checkmark \\
Mutable state       & $\times$ & $\times$ & $\times$ & \checkmark \\
Efficient verify    & \checkmark & \checkmark & $\times$ & $O(1)$ IVC$^{\dagger}$ \\
No trusted setup    & $\times$ & \checkmark & \checkmark & \checkmark \\
ASIC bound          & none & none & 8--16$\times$ & ${\sim}2{\times}$ \\
\bottomrule
\end{tabular*}
\label{tab:compare}
\raggedright\footnotesize $^{\dagger}$Design validated
(\S\ref{sec:ivc}, \S\ref{sec:concrete-ivc}); production
Binius implementation pending.
\end{table}

\posme{} fills this gap. Mutability is essential: in a static arena with independent
blocks, any block can be recomputed from the seed in $O(1)$,
so the adversary need not store the arena at all. A \emph{mutable}
arena forces storage because block values evolve
unpredictably under the ROM during execution. The persistent mutable arena \emph{is} the
computation state. Each step reads $d$ blocks at addresses
determined by the previous read's hash output (pointer chasing),
modifies the arena in-place, and chains a \emph{causal hash}
through the written block. Data and causal hash are
\emph{symbiotically bound}: neither can be fabricated without the
other. The bottleneck is DRAM random-access latency
(40--50\,ns on DDR5), with hash computation (${\sim}3$\,ns)
under 3\% of cost. HBM3 is only 1.3$\times$ faster than DDR5
at random access; GPU hardware is 14--19$\times$ \emph{slower}
than CPUs at this workload.

\textbf{Main Contributions.}
\begin{itemize}
    \item \textbf{Novel Construction:} We define \posme{}, the first
    primitive combining mutable arena state, data-dependent
    pointer-chase addressing, and symbiotic causal hash binding,
    formalized over a Boolean hypercube to enable binary-field
    folding.
    \item \textbf{Space-Time Security:} We prove a formal
    $S \cdot T = \Omega(K^2)$ lower bound for dynamic causal DAGs
    (Theorems~\ref{thm:st}--\ref{thm:stale}), extended to adaptive
    adversaries (Theorem~\ref{thm:adaptive}) with no factor loss.
    \item \textbf{ROM Validation:} We validate the ROM-uniformity
    assumption mathematically (Theorem~\ref{thm:coverage}) and
    empirically at the
    full recommended scale ($N{=}2^{24}$, $5.4{\times}10^8$
    reads, $\chi^2/\text{df} = 1.0004$).
    \item \textbf{Hardware Validation:} We validate ASIC resistance
    across 17 CPU and 4 GPU platforms, showing GPUs are
    \textbf{14--19$\times$ slower} and GPU throughput is choked
    by VRAM capacity.
    \item \textbf{IVC Design:} We provide a concrete Binius
    arithmetization (144{,}512 Boolean constraints per step)
    with analyzed fold cost and pipelining feasibility.
\end{itemize}

Section~\ref{sec:construction} presents the construction,
including the hypercube formulation and IVC design.
Section~\ref{sec:security} provides the security analysis
(Theorems~1--6). Section~\ref{sec:empirical} gives the
empirical validation. Section~\ref{sec:limitations}
discusses limitations.

\subsection{Related Work}

\textbf{PoSW.}~Cohen and Pietrzak~\cite{cohen2018posw} prove
traversal of a depth-robust graph via
Fiat-Shamir~\cite{fiatshamir1987} sampled Merkle proofs. \textbf{\posme{} differs by introducing a mutable arena}
(not a static DAG), where the access pattern is data-dependent
rather than fixed, and each node carries a causal hash binding its value
to its write history.

\textbf{MHFs.}~scrypt~\cite{percival2009}, proven maximally
memory-hard by~\cite{scrryptmaxmh2017}, introduced
memory-hardness; Argon2id~\cite{biryukov2016} achieves
bandwidth-hardness with ${\sim}2{\times}$ single-pass TMTO penalty
and $8$--$16{\times}$ ASIC advantage~\cite{ren2017}.
RandomX~\cite{tevador2019} targets CPU-egalitarianism via
random code execution, achieving $2$--$5{\times}$ ASIC bound.
\posme{} achieves higher TMTO resistance ($10{\times}$ at
$\rho{=}4$) and a tighter ${\sim}2{\times}$ ASIC bound via
latency-hardness rather than bandwidth- or compute-hardness.

\textbf{PoS/PoST.}~Proofs of Space~\cite{dziembowski2015}
and Space-Time~\cite{cohen2018posw} enforce storage over static
graphs. \posme{} extends to a mutable arena with per-block
temporal binding.
\textbf{CMC.}~Alwen et al.~\cite{alwen2017} formalized cumulative
memory complexity for static pebbling. \posme{}'s causal DAG
is dynamic, requiring the new framework of \S\ref{sec:pebbling}.

\section{Construction}
\label{sec:construction}

\subsection{Design Intuition}

Before presenting the formal algorithms, we outline the three
core design principles that differentiate \posme{} from existing
sequential primitives:

\begin{enumerate}
    \item \textbf{Data-Dependent Pointer Chasing:} The address of each
    subsequent read depends on the value of the current read. This
    enforces a strict sequential chain that cannot be
    parallelized.\footnote{The hardness of parallelizing pointer
    chasing has a long history in communication complexity;
    Nisan and Wigderson~\cite{nisanwigderson1991} initiated this
    line; Viola~\cite{viola2025pc} gives a recent tight
    $\Omega(n/k + k)$ bound for $k$-step pointer chasing.}
    \item \textbf{Mutable Arena State:} Unlike static graph-based
    proofs, the arena evolves at every step. Recomputing an
    evicted block requires replaying its write chain ($O(\rho)$
    hashes per miss; Theorem~\ref{thm:tmto}).
    \item \textbf{Symbiotic Causal Binding:} Every write binds the
    data value and its temporal write-history (causal hash)
    together. This ensures that an adversary cannot forge a single
    block without also forging its entire causal lineage.
\end{enumerate}

\subsection{Hypercube Arena and Multilinear Mapping}

We map the arena to a Boolean hypercube of dimension
$d_{\text{hc}}$, where $N = 2^{d_{\text{hc}}}$, addressing
each block by a binary coordinate vector
$v \in \{0, 1\}^{d_{\text{hc}}}$. The pointer chase becomes
a data-dependent walk across hypercube vertices. The arena
state at step $t$ is expressible as a multilinear polynomial
$P_t(x_1, \dots, x_{d_{\text{hc}}})$ over $\mathbb{F}_2$
(or an extension tower), and each step's writes are low-degree
evaluations. This imposes no overhead on the physical Prover
but provides the foundation for $O(1)$ verification via
binary-field folding (\S\ref{sec:ivc}).

\subsection{Formal Definition}

We formalize \posme{} as a triple of algorithms $(\mathsf{Gen},
\mathsf{Prove}, \mathsf{Verify})$ operating over an arena $A$ of
$N$ blocks.

\begin{itemize}
    \item $\mathsf{Gen}(s, N, K) \to (T_K, \{r_t\})$. The Prover
    initializes $A$ via $\textsc{Init}(s, N)$ and executes $K$ steps of
    $\textsc{Step}$ to produce the final transcript $T_K$ and root
    sequence.
    \item $\mathsf{Prove}(T_K, \{r_t\}, A, Q, R) \to \Pi$. The Prover
    commits to the root sequence and generates $Q$ witnesses of depth
    $R$ for Fiat-Shamir challenged steps.
    \item $\mathsf{Verify}(\Pi, s, N, K, Q, R) \to \{0, 1\}$. The
    Verifier re-derives challenges and checks witnesses against the
    initial seed and committed roots.
\end{itemize}

\textbf{Arena isolation.} The seed $s$ \textbf{must} bind a
unique task identifier (nonce, block height, challenge, etc.):
$s = H(\texttt{task\_id} \| \texttt{nonce})$. Because
$\textsc{Init}$ is deterministic, distinct seeds produce
statistically independent arenas. This prevents an adversary
from amortizing a single arena across multiple proof instances
and is essential to the capacity-bandwidth bound analyzed in
\S\ref{sec:throughput}.

\begin{proposition}[Completeness]
For any seed $s$ and parameters $(N, K, Q, R)$, an honest Prover
following the protocol will always produce a proof $\Pi$ such that
$\mathsf{Verify}(\Pi, s, N, K, Q, R) = 1$.
\end{proposition}

\begin{proof}
All algorithms are deterministic; an honest Prover computes
roots and transcripts correctly, and all checks in
Algorithm~\ref{alg:verify} compare honest values against
themselves.
\end{proof}

\subsection{Proof Size and Complexity}

The proof size $|\Pi|$ is dominated by the $Q$ witnesses. Each
witness $\pi_i$ contains openings for the challenged step and its
provenance ancestors to depth $R$.

\begin{proposition}[Proof Size]
For a recursion depth $R$, the number of blocks opened per challenge
is $B = \sum_{\ell=0}^{R-1} d^\ell(d+1)$. The total proof size is
approximately $Q \cdot B \cdot (\lambda + \log N \cdot |H|)$ bytes.
\end{proposition}

For the recommended parameters ($N{=}2^{24}$, $d{=}8$,
$Q{=}128$), a proof at $R{=}2$ opens $B{=}81$ blocks
(${\sim}8.5$\,MB). At $R{=}3$, $B{=}657$ blocks
(${\sim}70$\,MB). These proofs are post-quantum secure,
require no trusted setup, and are generated autonomously
during execution.

\subsection{Proof Size Optimization}

The recursion depth $R$ and challenge count $Q$ present a direct
tradeoff between security margin and proof size. Table~\ref{tab:proof-size}
provides concrete MiB-per-proof costs for implementers.

\vspace{4pt}
\noindent
\begin{minipage}{\columnwidth}
\centering
\refstepcounter{table}
\small\textsc{Table \thetable}\\
\textsc{Proof size (MiB) for $N=2^{24}$ (1\,GiB arena).}
\vspace{2pt}

\begin{tabular*}{\columnwidth}{@{\extracolsep{\fill}}lrrr@{}}
\toprule
Recursion ($R$) & Challenges ($Q$) & Blocks ($B$) & Size (MiB) \\
\midrule
2 & 64 & 81 & 3.9 \\
2 & 128 & 81 & 7.9 \\
\textbf{3} & \textbf{64} & \textbf{657} & \textbf{32.1} \\
3 & 128 & 657 & 64.2 \\
\bottomrule
\end{tabular*}
\label{tab:proof-size}
\end{minipage}
\vspace{4pt}

While $R=3$ yields larger proofs, it provides exponentially
higher fabrication resistance by checking the witnesses of the
writers' writers. For bandwidth-constrained environments (e.g.,
light clients), $R=2$ with $Q=128$ offers a compact
$\sim$8\,MB proof while maintaining high confidence.

\FloatBarrier
\subsection{Arena and Initialization}

The arena consists of $N = 2^{d_{\text{hc}}}$ blocks on a Boolean
hypercube, each storing $(\texttt{data}, \texttt{causal})$, each
$\lambda{=}256$ bits, initialized deterministically from seed $s$
via Algorithm~\ref{alg:init}. Vertices are enumerated in standard
binary order; the skip-link parent of vertex $v$ is
$v \gg 1$ (right-shift by one bit), creating a binary DAG
requiring $\Omega(\sqrt{N})$ space to evaluate any single vertex.

\newcommand{\algcmt}[1]{\hfill\makebox[2.4cm][l]{\normalfont$\triangleright$\,\textit{\footnotesize#1}}}
\begin{algorithm}[t]
\caption{\textsc{Init}$(s, N)$: hypercube initialization}\label{alg:init}
\begin{algorithmic}[1]
\Require Seed $s$, arena dimension $d_{\text{hc}}$ with $N = 2^{d_{\text{hc}}}$
\Ensure Arena $A[v]$ for $v \in \{0,1\}^{d_{\text{hc}}}$, Merkle root $r_0$, transcript $T_0$
\Statex
\State $A[\mathbf{0}].d \gets H(\text{``init''} \| s \| 0)$ \algcmt{seed data}
\State $A[\mathbf{0}].h \gets H(\text{``causal''} \| s \| 0)$ \algcmt{seed causal}
\For{$i \gets 1$ \textbf{to} $N - 1$} \algcmt{binary order}
  \State $v \gets \text{binary}(i, d_{\text{hc}})$;\; $v_p \gets v \gg 1$ \algcmt{skip-link}
  \State $A[v].d \gets H(\text{``init''} \| s \| i \| A[v_p].d)$
  \State $A[v].h \gets H(\text{``causal''} \| s \| i \| A[v_p].h)$
\EndFor
\Statex
\State $r_0 \gets \textsc{MerkRoot}(A)$ \algcmt{commit all vertices}
\State $T_0 \gets H(s \| N \| r_0)$ \algcmt{init transcript}
\State \Return $(A, r_0, T_0)$
\end{algorithmic}
\end{algorithm}

\subsection{Step Function}

Each step $t$ executes Algorithm~\ref{alg:step}. The coordinate
$v_{j+1}$ depends on the block at $v_j$; the $d$ reads are
strictly sequential. Symbiotic binding creates bidirectional
dependency: new data incorporates the old causal hash, and the
new causal hash incorporates the cursor. Addressing uses binary
projection: the first $d_{\text{hc}}$ bits of the hash output
yield a coordinate vector $v \in \{0,1\}^{d_{\text{hc}}}$,
selecting a vertex on the Boolean hypercube.

\begin{algorithm}[t]
\caption{\textsc{Step}$(t, A, T_{t-1})$: hypercube edge walk}\label{alg:step}
\begin{algorithmic}[1]
\Require Step index $t \in [1, K]$, arena $A$ of dimension $d_{\text{hc}}$, transcript $T_{t-1}$
\Ensure Updated arena $A$, transcript $T_t$, Merkle root $r_t$
\Statex
\State $c \gets T_{t-1}$ \algcmt{init cursor}
\For{$j \gets 0$ \textbf{to} $d - 1$} \algcmt{$d$ sequential reads}
  \State $v_j \gets \text{first } d_{\text{hc}} \text{ bits of } H(\text{``addr''} \| c \| j)$ \algcmt{hypercube coord}
  \State $(d_j, h_j) \gets A[v_j]$ \algcmt{vertex read}
  \State $c \gets H(c \| d_j \| h_j)$ \algcmt{chain cursor}
\EndFor
\Statex
\State $v_w \gets \text{first } d_{\text{hc}} \text{ bits of } H(\text{``write''} \| c)$ \algcmt{write coord}
\State $(d_w, h_w) \gets A[v_w]$ \algcmt{read old vertex}
\State $A[v_w].d \gets H(d_w \| c \| h_w)$ \algcmt{symbiotic bind}
\State $A[v_w].h \gets H(h_w \| c \| t)$ \algcmt{causal bind}
\Statex
\State $r_t \gets \textsc{Merk.Upd}(r_{t-1}, v_w, A[v_w])$ \algcmt{update root}
\State $T_t \gets H(T_{t-1} \| t \| c \| r_t)$ \algcmt{extend transcript}
\State \Return $(T_t, r_t)$
\end{algorithmic}
\end{algorithm}

\subsection{Roles of the Causal Hash}\label{sec:causal-roles}

The causal hash field $h$ serves three distinct purposes in
\posme{}. We enumerate them here to prevent conflation in the
security analysis that follows.

\begin{enumerate}
    \item \textbf{Operational binding (Construction, \S\ref{sec:construction}).}
    Each block's causal hash chains its write history into the block
    value, creating a per-block temporal lineage. The symbiotic bind
    (Algorithm~\ref{alg:step}, line~9) ensures data and causal hash
    are mutually dependent: neither can be computed without the other.

    \item \textbf{Soundness amplification (Theorem~\ref{thm:soundness}, \S\ref{sec:security}).}
    Because the causal hash is folded into the cursor at every read
    (Algorithm~\ref{alg:step}, line~5), any forged block produces a
    divergent cursor, which propagates into a divergent transcript
    $T_t$. This tightens the reduction to collision resistance: the
    adversary must forge \emph{both} data and causal lineage
    simultaneously.

    \item \textbf{TMTO constant factor (Theorem~\ref{thm:tmto}, \S\ref{sec:tmto}).}
    When the adversary discards a block and must recompute it on
    a cache miss, the write chain must be traversed for both the
    data field and the causal field, doubling the per-miss
    recomputation cost ($2\rho$ vs.\ $\rho$). This is a
    $2{\times}$ constant, not an asymptotic improvement.
\end{enumerate}

\noindent These roles are complementary but logically independent:
operational binding is a construction choice, soundness
amplification is a proof-theoretic consequence, and the TMTO
constant is an empirical cost multiplier.

\subsection{Classical Verification (Merkle/Fiat-Shamir)}
\label{sec:classical-verify}

\posme{} supports two verification modes. The \emph{classical}
mode (Algorithms~\ref{alg:prove}--\ref{alg:verify}) uses
Fiat-Shamir~\cite{fiatshamir1987} sampled
Merkle~\cite{merkle1988} witnesses and requires only a hash
function; it is post-quantum secure and requires no trusted
setup. The \emph{IVC} mode (\S\ref{sec:ivc},
Algorithms~\ref{alg:ivc-prove}--\ref{alg:ivc-verify}) achieves
$O(1)$ proof size via binary-field folding but introduces an
additional cryptographic assumption (polynomial commitment
soundness).

In the classical mode, challenges are derived after the Prover
commits to the complete root sequence, preventing selective
disclosure.

\begin{algorithm}[t]
\caption{\textsc{Prove}$(T_K, \{r_t\}, A, Q, R)$}\label{alg:prove}
\begin{algorithmic}[1]
\Require $T_K$, roots $\{r_0, \ldots, r_K\}$, arena $A$, $Q$, $R$
\Ensure $\Pi = (T_K, C, \{\pi_i\}_{i \leq Q})$
\Statex
\State $C \gets \textsc{MerkRoot}(r_0, \ldots, r_K)$ \algcmt{commit roots}
\State $\sigma \gets H(T_K \| C)$ \algcmt{Fiat-Shamir}
\For{$i \gets 1$ \textbf{to} $Q$}
  \State $s_i \gets H(\sigma \| i) \bmod K$ \algcmt{select step}
  \State $\pi_i.\mathrm{wit} \gets (T_{s_i{-}1}, \{v_j,d_j,h_j\}, v_w)$ \algcmt{witness}
  \State $\pi_i.\mathrm{pre} \gets \textsc{M.Pf}(r_{s_i{-}1}, \{v_j,v_w\})$ \algcmt{pre-state}
  \State $\pi_i.\mathrm{post} \gets \textsc{M.Pf}(r_{s_i}, v_w)$ \algcmt{post-state}
  \State $\pi_i.\mathrm{chain} \gets \textsc{M.Pf}(C, s_i{-}1, s_i)$ \algcmt{root pair}
  \State $\pi_i.\mathrm{prov} \gets \textsc{Prov}(\{v_j\}, v_w, R)$ \algcmt{depth $R$}
\EndFor
\State \Return $(T_K, C, \{\pi_1, \ldots, \pi_Q\})$
\end{algorithmic}
\end{algorithm}

\begin{algorithm}[t]
\caption{\textsc{Verify}$(T_K, C, \{\pi_i\}, N, d, K, Q, R)$.
Failed \textbf{check} $\Rightarrow$ \textsc{reject}.}\label{alg:verify}
\begin{algorithmic}[1]
\State $\sigma \gets H(T_K \| C)$ \algcmt{Fiat-Shamir}
\For{$i \gets 1$ \textbf{to} $Q$}
  \State $s_i \gets H(\sigma \| i) \bmod K$ \algcmt{re-derive step}
  \State $(r^-,\, r^+) \gets$ roots from $\pi_i.\mathrm{chain}$
  \State \textbf{check} $\textsc{M.Vfy}(C, s_i{-}1, r^-)$ \algcmt{root $r^- \in C$}
  \State \textbf{check} $\textsc{M.Vfy}(C, s_i, r^+)$ \algcmt{root $r^+ \in C$}
  \Statex
  \State \textbf{check} reads in $r^-$ via $\pi_i.\mathrm{pre}$ \algcmt{pre-state}
  \State $(T',\, r') \gets \textsc{Step}(s_i,\, \pi_i.\mathrm{wit})$ \algcmt{replay step}
  \State \textbf{check} write in $r^+$ via $\pi_i.\mathrm{post}$ \algcmt{post-state}
  \State \textbf{check} $r' = r^+$ \algcmt{root matches}
  \Statex
  \For{$\ell \gets 1$ \textbf{to} $R$}
    \For{\textbf{each} $b$ at depth $\ell$ in $\pi_i.\mathrm{prov}$}
      \State \textbf{check} writer $\wedge$ Merkle proof \algcmt{provenance}
    \EndFor
  \EndFor
\EndFor
\State \Return \textsc{accept}
\end{algorithmic}
\end{algorithm}

\textbf{Classical verification cost:} $O(Q \cdot d^R \cdot \log N)$
hashes (${\sim}6$\,ms desktop, 60--300\,ms mobile). No arena
allocation. For $O(1)$ verification, see \S\ref{sec:ivc}.

\subsection{Constant-Size Verification via Binary-Field Folding}
\label{sec:ivc}

A fundamental limitation of legacy memory-hard functions and early
\posme{} verification strategies is that proof size scales with
sequence length $K$ or requires $O(Q \cdot d^R)$ Merkle path
verifications. By explicitly structuring the arena as a Boolean
hypercube (\S\ref{sec:construction}), the construction natively
integrates with recent breakthroughs in binary-field succinct
arguments, notably Binius~\cite{binius2024}.

Because Binius operates over towers of binary fields and utilizes
multilinear polynomial commitments, the \posme{} step function
maps directly to the prover's arithmetic circuit without the
heavy overhead of bit-decomposition required by prime-field
SNARKs. This enables Incrementally Verifiable Computation
(IVC)~\cite{nova2022}: the Prover ``folds'' the proof of
step $t$ into
the proof of step $t{-}1$. Upon completing the $K$-th step, the
final proof is already generated. This reduces the proof size
from tens of megabytes (Table~\ref{tab:proof-size}) to a constant
size of a few kilobytes, achieving $O(1)$ verification regardless
of the number of steps $K$ or the recursion depth $R$.

\textbf{Compatibility with BLAKE3.} Binius was specifically
designed to efficiently verify standard bitwise operations
(XOR, rotation, addition mod $2^{32}$), making it natively
compatible with BLAKE3's compression function. No change to the
hash primitive is required to benefit from binary-field folding.

\textbf{IVC Prove and Verify.}
In IVC mode, the Prover maintains a folding accumulator $U_t$
(Algorithm~\ref{alg:ivc-prove}). Each step folds the state
transition into $U_t$ via the operator $\mathcal{F}$. The
Verifier (Algorithm~\ref{alg:ivc-verify}) checks only the
final accumulator, achieving $O(1)$ verification independent
of $K$.

\begin{algorithm}[t]
\caption{\textsc{IVC-Prove}$(s, N, K)$: folded proof generation}\label{alg:ivc-prove}
\begin{algorithmic}[1]
\Require Seed $s$, arena dimension $d_{\text{hc}}$, steps $K$
\Ensure Folded proof $\Pi_{\text{IVC}} = (T_K, U_K)$
\Statex
\State $(A, r_0, T_0) \gets \textsc{Init}(s, N)$
\State $U_0 \gets \bot$ \algcmt{empty accumulator}
\For{$t \gets 1$ \textbf{to} $K$}
  \State $(T_t, r_t) \gets \textsc{Step}(t, A, T_{t-1})$ \algcmt{execute step}
  \State $U_t \gets \mathcal{F}(U_{t-1},\, T_{t-1},\, T_t,\, r_t)$ \algcmt{fold step into proof}
\EndFor
\State \Return $(T_K, U_K)$
\end{algorithmic}
\end{algorithm}

\begin{algorithm}[t]
\caption{\textsc{IVC-Verify}$(s, N, K, T_K, U_K)$.
Failed \textbf{check} $\Rightarrow$ \textsc{reject}.}\label{alg:ivc-verify}
\begin{algorithmic}[1]
\Require Seed $s$, parameters $N, K$, transcript $T_K$, accumulator $U_K$
\Statex
\State \textbf{check} $U_K$ is a valid Binius opening \algcmt{proof well-formed}
\State \textbf{check} $\mathcal{F}.\textsc{Verify}(U_K,\, s,\, N,\, K,\, T_K)$ \algcmt{$O(1)$ check}
\State \Return \textsc{accept}
\end{algorithmic}
\end{algorithm}

\textbf{IVC verification cost:} $O(1)$ field operations plus a
constant number of hash evaluations, independent of $K$.
Proof size is a few kilobytes. The classical mode
(\S\ref{sec:classical-verify}) remains available as a fallback
requiring only hash function security.

\subsection{Concrete Arithmetization and Pipelining}
\label{sec:concrete-ivc}

To achieve the $O(1)$ verification of \S\ref{sec:ivc} without
violating the latency bound, the \posme{} step function must be
efficiently arithmetized over the binary-field tower.

\textbf{BLAKE3 Boolean trace.}
BLAKE3 operates on 32-bit words using three operations:
XOR ($\oplus$), right rotation ($\ggg$), and addition modulo
$2^{32}$ ($\boxplus$). In a prime-field SNARK, these require
bit-decomposition (${\sim}32$ range-check constraints per
32-bit operation). In the Binius $\mathbb{F}_2$ multilinear trace,
XOR is native field addition and rotation is deterministic
wire-routing; both incur \textbf{zero non-linear constraints}.
The only non-trivial arithmetization arises from the carry
bits in $\boxplus$, represented via a low-degree
carry-lookahead polynomial over $\mathbb{F}_2$.
Consequently, the computational trace of the step function
requires only a few thousand Boolean variables, minimizing the
folding Prover's sumcheck overhead.

\textbf{Log-derivative memory checking.}
To verify the multilinear memory evaluations without Merkle
paths in the circuit, we employ log-derivative memory
checking~\cite{logderivlookup2022} over the
hypercube~\cite{binius2024}. Each read and write at
coordinate $v \in \{0,1\}^{d_{\text{hc}}}$ updates a running
fractional sum in an extension field (e.g.,
$\mathbb{F}_{2^{128}}$). The step circuit verifies only
the $O(1)$ update to the running sum; the global memory
consistency check (matching Read Set and Write Set
permutation fingerprints) is deferred to the final Verifier.
This replaces $O(\log N)$ Merkle constraints per access with
a single finite-field addition.

\textbf{Prover pipelining and fold cost.}
Let $T_{\text{mem}} \approx 50\,\text{ns}$ be the DRAM
random-access latency and $T_{\text{fold}}$ the time to
compute the Binius IVC fold for a single step. The step
function requires $2d{+}4 = 20$ BLAKE3 evaluations. Each
BLAKE3 compression has 7 rounds $\times$ 8 G-functions
$\times$ 4 additions mod $2^{32}$ $\times$ 32 carry bits
$= 7{,}168$ AND gates. With 20 compressions plus
${\sim}1{,}152$ gates for log-derivative memory checking,
the total is $\mathbf{144{,}512}$ Boolean constraints
per step.

With deferred commitment (standard in IVC; the polynomial
commitment is generated only for the final proof, not
per-step), the per-step cost is the Binius sumcheck alone.
Over $\lceil \log_2 144{,}512 \rceil = 18$ rounds, the
sumcheck evaluates ${\sim}2.6 \times 10^6$ $\mathbb{F}_2$
operations. Because $\mathbb{F}_2$ arithmetic (XOR, AND)
is SIMD-vectorizable, a single AVX-512 instruction (512
bits) processes 512 $\mathbb{F}_2$ operations in parallel:

\[
T_{\text{fold}} =
  \frac{2.6 \times 10^6}{512} \times
  \frac{1}{f_{\text{CPU}}}
  \approx \frac{5{,}078\;\text{instr.}}{4\;\text{GHz}}
  = 1{,}270\;\text{ns}.
\]

\noindent The per-step DRAM cost is
$T_{\text{mem}} = (d{+}1) \times 50\,\text{ns} = 450\,\text{ns}$.
Since $T_{\text{fold}} / T_{\text{mem}} \approx 2.8$,
\textbf{three fold threads} (4 threads total on AVX-512)
achieve aggregate fold throughput of $423\,\text{ns}$,
keeping pace with the pointer-chase. On consumer AVX2
hardware (256-bit SIMD),
$T_{\text{fold}} \approx 2{,}540\,\text{ns}$, requiring 6
fold threads.

Critically, this parallelism is \emph{asymmetric}: the
adversary cannot exploit it because the pointer-chase
remains strictly single-threaded, and GPUs are
14--19${\times}$ slower at the sequential bottleneck
(\S\ref{sec:empirical}). The fold parallelism benefits
only the honest Prover.

\section{Security Analysis}
\label{sec:security}

\subsection{Threat Model}

The adversary is PPT with ROM access to
$H$~\cite{alwenserbinenko2015}, receiving $(s, N, K, d, Q, R)$.
Goals: (1) forgery ($T_K' \neq T_K$), or (2) space reduction
(storage $< N \cdot B$). Custom hardware permitted.

\subsection{Forgery Prevention}

\begin{theorem}[Soundness]\label{thm:soundness}
Any adversary producing $(T_K', C_{\text{roots}}', \pi')$ with
$T_K' \neq T_K$ that passes verification has advantage at most
$K \cdot \varepsilon_{\text{cr}}$.
\end{theorem}

\begin{proof}[Proof sketch]
Let $t^{*}$ be the first step where the adversary's transcript
diverges. Then
$T_{t^{*}} = H(T_{t^{*}-1} \| t^{*} \| \texttt{cursor} \|
\texttt{root}_{t^{*}})$ was computed with at least one differing
input. Matching the honest output requires finding a collision
in $H$. Union bound over $K$ steps gives
$K \cdot \varepsilon_{\text{cr}}$.
\end{proof}

Causal hashes additionally double write-chain traversal cost
($2\rho$ vs $\rho$ per miss), a $2{\times}$ constant factor
(see \S\ref{sec:causal-roles}, Role~3).
Their primary contribution is \emph{soundness}, not TMTO
amplification.

\subsection{TMTO Lower Bound}\label{sec:tmto}

\textbf{Write density constraint.}
Define the \emph{write density} $\rho = K/N$, the expected number
of times each arena block is overwritten during a $K$-step
execution. When $\rho < 1$, a constant fraction of blocks are
never written; the adversary can recompute them from the
initialization seed at $O(1)$ cost per block, defeating the
storage requirement. Meaningful TMTO resistance therefore requires:

\begin{equation}
\rho = K/N \geq 1 \qquad (K \geq N).
\end{equation}

\noindent We recommend $\rho \geq 4$; see the discussion after
Theorem~\ref{thm:tmto}.

\begin{theorem}[TMTO]\label{thm:tmto}
An adversary storing $\alpha N$ blocks ($0 \leq \alpha < 1$) and
all $K$ cursors performs expected computation
$T_{\text{adv}} \geq K d \bigl(1 + (1{-}\alpha)(2\rho{+}1)\bigr)$.
\end{theorem}

\begin{proof}[Proof sketch]
Fix a step $t \in [1, K]$. The $d$ read coordinates are derived as
$v_j = \text{first } d_{\text{hc}} \text{ bits of } H(\text{``addr''} \| c \| j)$,
where $c$ depends on prior reads. In the ROM, each $v_j$ is
uniform over $\{0,1\}^{d_{\text{hc}}}$, independent of the
adversary's stored set $S$ (since $S$ was chosen before the ROM
queries). Each read hits $S$ with probability $\alpha$ and misses
with probability $1 - \alpha$.

On a miss, the adversary must reconstruct block $A[v_j]$ by
replaying its write chain: the block was overwritten $\rho$ times
in expectation, each overwrite requiring one hash for the data
field and one for the causal field, plus the final read itself,
totaling $2\rho + 1$ hash evaluations per miss. (The $2{\times}$
factor from causal hashes is the constant described in
\S\ref{sec:causal-roles}, Role~3.)

Summing over $K$ steps of $d$ reads each:
\[
T_{\text{adv}} \geq \sum_{t=1}^{K} \sum_{j=0}^{d-1}
  \bigl[\alpha + (1{-}\alpha)(2\rho{+}1)\bigr]
  = Kd\bigl(1 + (1{-}\alpha)(2\rho{+}1)\bigr).
\]
The bound is tight when the adversary stores the optimal $\alpha N$
blocks (those most frequently accessed); non-uniform strategies
cannot improve the expectation because ROM-uniform addressing
makes all blocks equiprobable targets.
\end{proof}

\textbf{Why $\rho \geq 4$.} At $\rho = 4$ and $\alpha = 0$
(no storage), the TMTO penalty is $T_{\text{adv}} / T_{\text{honest}}
= 1 + (2 \cdot 4 + 1) = 10{\times}$. This exceeds the measured
ASIC advantage (${\sim}2{\times}$) by $5{\times}$: a
storage-reducing adversary is strictly worse off than one with
faster hardware. Verification cost is independent of $\rho$
(it depends on $Q$, $d$, $R$, $N$ only), so increasing $\rho$
adds Prover wall-clock time with no Verifier penalty. The bound
assumes optimal cursor storage; sparse cursors strictly increase
adversary cost.

\begin{figure}[b]
\centering
\begin{tikzpicture}
\begin{axis}[
  width=0.9\columnwidth, height=4.5cm,
  xbar,
  bar width=8pt,
  symbolic y coords={PoSME GPU,PoSME ASIC,Argon2id,scrypt,SHA-256},
  ytick=data,
  y tick label style={font=\scriptsize},
  xlabel={Hardware advantage ($\times$ consumer CPU)},
  xlabel style={font=\footnotesize},
  xmin=0.03, xmax=15000,
  xmode=log,
  log ticks with fixed point,
  grid=major, grid style={gray!20},
]
\addplot[fill=green!60] coordinates
  {(0.07,PoSME GPU) (2,PoSME ASIC) (16,Argon2id) (100,scrypt) (10000,SHA-256)};
\end{axis}
\end{tikzpicture}
\caption{Hardware advantage over consumer CPU (log scale).
\posme{} resists acceleration; existing primitives do not.}
\label{fig:asic-compare}
\end{figure}
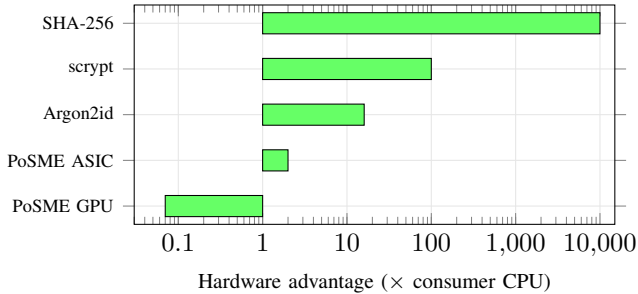

\subsection{Dynamic Pebbling}\label{sec:pebbling}

\posme{}'s causal DAG is \emph{dynamic}: edges are created during
execution at ROM-uniform targets, so the graph topology is not
known until the computation completes. This distinguishes it from
the static graphs analyzed by Alwen et al.~\cite{alwen2017} and
Boneh et al.~\cite{boneh2016balloon}, where the edge set is fixed
at construction time and classical pebbling lower bounds apply
directly.

\textbf{The obstacle.}
Static-graph pebbling proofs (e.g., cumulative memory complexity
bounds for Argon2~\cite{alwen2017}) rely on two properties:
(1)~the graph is depth-robust, meaning no small set of nodes
separates all long paths; and (2)~the graph topology is known to
the analyst a priori, enabling combinatorial arguments over its
structure. In \posme{}, neither property holds in the standard
sense. The edge from step $t$'s write target back to its read
sources is determined by the ROM output projected onto
$\{0,1\}^{d_{\text{hc}}}$; the
analyst cannot enumerate paths without fixing the ROM, which
couples the combinatorial argument to a specific oracle
instantiation.
Recent work by Blocki and
Holman~\cite{blocki2025ddmhf}, extending the sustained-space
framework of~\cite{sustainedspace2018}, establishes
sustained-space tradeoffs for \emph{data-dependent}
MHFs in the parallel ROM, proving that any dynamic pebbling
strategy either maintains $\Omega(N)$ memory for $\Omega(N)$
steps or incurs cumulative cost $\Omega(N^{2.5-\varepsilon})$.
Their model, which allows edges to be revealed online,
directly informs our two-phase decomposition below.

\textbf{Two-phase decomposition.}
We resolve this obstacle by decomposing \posme{}'s computation
graph into a \emph{static backbone} (the Init DAG, whose
structure is known a priori) and a \emph{dynamic overlay} (the
Step writes, whose edges are revealed online). The static
backbone provides a deterministic termination condition for
recomputation; the dynamic overlay provides an exponentially
growing recursive miss penalty.

\begin{theorem}[Space-Time Product]\label{thm:st}
In the ROM, any adversary storing $S = \alpha N$ vertices
($0 < \alpha < 1$) requires expected computation $T$ satisfying:
\[
S \cdot T \;\geq\; \frac{\alpha(1{-}\alpha)\, d\,
  W(\alpha,\rho)}{\rho} \cdot K^2
\]
where $W(\alpha,\rho) = \sum_{\ell=0}^{\rho}
[d(1{-}\alpha)]^{\ell}$ is the expected recursive
recomputation cost per cache miss.
For the recommended parameters ($d{=}8$, $\rho{=}4$),
$S \cdot T = \Omega(K^2)$ for all constant $\alpha$.
\end{theorem}

\begin{proof}[Proof sketch]
\textbf{Phase~1: Static backbone (termination).}
The Init DAG has skip-link edges $v \to v \gg 1$, forming a
binary tree of depth $d_{\text{hc}}$. Reconstructing any
Init-state vertex requires traversing its skip-link ancestry:
$O(\log N)$ hash evaluations. This provides a deterministic
floor for any recomputation that reaches a vertex still in its
initialization state.

\textbf{Phase~2: Recursive cascade (dynamic overlay).}
Consider a cache miss at step $t$: the adversary must
reconstruct vertex $v_j$, which was last written at step
$t' < t$. That write read $d$ vertices, each missing
independently with probability $(1{-}\alpha)$ in the ROM
(by the same uniformity argument as
Theorem~\ref{thm:tmto}). Each secondary miss triggers
further recomputation, forming a Galton-Watson branching
process with expected offspring $d(1{-}\alpha)$.

\textbf{Depth bound.}
Under ROM-uniform writes (one write per step to a uniform
random vertex), the time since the last write to any given
vertex is geometric with success probability $1/N$,
giving expected gap $N$ and standard deviation
$\sqrt{N(N{-}1)} \approx N$. Each recursion level
regresses the timeline by one such gap (${\sim}N$ steps).
After $\rho = K/N$ levels, the cumulative regression
covers ${\sim}\rho N = K$ steps, reaching the Init state,
where Phase~1 terminates the branch at $O(\log N)$ cost. The expected number of live misses at
depth $\ell$ is $[d(1{-}\alpha)]^{\ell}$, giving total
expected work per primary miss:
\[
W(\alpha,\rho) = \sum_{\ell=0}^{\rho}
  [d(1{-}\alpha)]^{\ell}.
\]
When $d(1{-}\alpha) > 1$ (i.e., $\alpha < 1{-}1/d$),
$W$ grows exponentially in $\rho$.

\textbf{Product.}
Each of the $K$ steps performs $d$ reads, each missing
with probability $(1{-}\alpha)$ and costing $W$ expected
hash evaluations. Total: $T \geq K d (1{-}\alpha) W$.
Multiplying by $S = \alpha N$ and substituting
$K = \rho N$:
\[
S \cdot T \;\geq\;
  \frac{\alpha(1{-}\alpha)\, d\, W(\alpha,\rho)}{\rho}
  \cdot K^2.
\]
\end{proof}

\textbf{Concrete evaluation ($d{=}8$, $\rho{=}4$).}
Table~\ref{tab:st-product} shows the space-time product
at representative storage fractions. The critical threshold
is $\alpha_c = 1 - 1/d = 7/8$: below this, the branching
process is supercritical and $W$ grows exponentially.
Even at $\alpha_c$ the bound exceeds $K^2$.

\begin{table}[H]
\centering
\caption{Space-time product $S \cdot T / K^2$ at $d{=}8$,
$\rho{=}4$ for varying storage fraction $\alpha$.}
\begin{tabular*}{\columnwidth}{@{\extracolsep{\fill}}rrrrr@{}}
\toprule
$\alpha$ & $d(1{-}\alpha)$ & $W$ & $S{\cdot}T/K^2$ & Regime \\
\midrule
$1/6$  & 6.67 & 2324 & 645 & supercritical \\
$1/4$  & 6    & 1555 & 583 & supercritical \\
$1/2$  & 4    & 341  & 171 & supercritical \\
$3/4$  & 2    & 31   & 11.6 & supercritical \\
$7/8$  & 1    & 5    & 1.09 & critical \\
\bottomrule
\end{tabular*}
\label{tab:st-product}
\end{table}

\noindent At the adversary-optimal storage $\alpha^{*} =
1/(\rho{+}2) = 1/6$, the space-time product is $645\,K^2$,
over $300{\times}$ the honest cost of $2\,K^2$.
Even storing $7/8$ of the arena (a mere 12.5\% reduction)
yields $S \cdot T \geq 1.09\,K^2$. Any meaningful storage
reduction triggers exponential recomputation.

\textbf{Strengthened bound: temporal staleness.}
Theorem~\ref{thm:st} conservatively assumes the adversary's
cache hit rate is $\alpha$ at every recursion depth.
This overstates the adversary's capability. The adversary
stores \emph{current} vertex states (at execution step $t$),
but reconstruction requires \emph{historical} states (at
step $t' < t$). Under ROM-uniform writes, vertex $u$ is
rewritten in the interval $(t', t]$ with probability
$1 - e^{-(t-t')/N}$, rendering the stored value stale.
Since each recursion level regresses ${\sim}N$ steps,
the effective hit rate at depth $\ell$ decays exponentially.

\begin{theorem}[Temporal Staleness]\label{thm:stale}
Under the conditions of Theorem~\ref{thm:st}, the
effective cache hit rate at recursion depth $\ell$ is
at most $\alpha e^{-\ell}$, and the expected recomputation
cost per miss satisfies:
\[
W^{*}\!(\alpha,\rho) \;\geq\;
  \sum_{\ell=0}^{\rho}\;
  \prod_{k=0}^{\ell-1} d\bigl(1 - \alpha e^{-k}\bigr).
\]
For the recommended parameters ($d{=}8$, $\rho{=}4$), the
branching factor $d(1 - \alpha e^{-k}) > 1$ for all
$k \geq 1$ and all $\alpha < 1$.
\end{theorem}

\begin{proof}[Proof sketch]
The adversary maintains $\alpha N$ current vertex states.
At recursion depth $\ell$, reconstruction needs vertex $u$'s
state from ${\sim}\ell N$ steps in the past. Under
ROM-uniform writes (one write per step to a uniform random
vertex), the probability that $u$ was \emph{not} rewritten in
$\ell N$ steps is $(1 - 1/N)^{\ell N} \approx e^{-\ell}$.
The stored state is valid only if $u$ was not rewritten, so
the effective hit rate is $\alpha e^{-\ell}$.

The expected offspring at depth $\ell$ is
$d(1 - \alpha e^{-\ell})$. For $d = 8$ and $\ell \geq 1$:
$d(1 - \alpha e^{-1}) = 8(1 - \alpha/e) > 1$ for all
$\alpha < e(1{-}1/8) = 2.38$, which holds for any
$\alpha \in [0, 1)$. Thus the branching process is
supercritical at every depth beyond the first level,
regardless of $\alpha$.
\end{proof}

\textbf{Elimination of the critical threshold.}
Under Theorem~\ref{thm:st}, the critical storage fraction
$\alpha_c = 1{-}1/d = 7/8$ separates the supercritical
and subcritical regimes. Temporal staleness eliminates this
threshold entirely: even at $\alpha = 7/8$, the level-1
branching factor is $d(1{-}\alpha/e) = 5.42$, driving the
recomputation cost from $W = 5$ to $W^{*} = 338$.
Table~\ref{tab:st-strengthened} compares both bounds.

\begin{table}[H]
\centering
\caption{Temporal staleness strengthening ($d{=}8$,
$\rho{=}4$). $W$ = basic miss cost (Thm.~\ref{thm:st}),
$W^{*}$ = strengthened (Thm.~\ref{thm:stale}).}
\begin{tabular*}{\columnwidth}{@{\extracolsep{\fill}}rrrrrr@{}}
\toprule
$\alpha$ & $W$ & $W^{*}$ &
  $\frac{S{\cdot}T}{K^2}$ &
  $\frac{S{\cdot}T}{K^2}\!(W^{*})$ & Gain \\
\midrule
$1/6$  & 2324 & 3555 & 645 & 987  & $1.5{\times}$ \\
$1/2$  & 341  & 1746 & 171 & 873  & $5.1{\times}$ \\
$7/8$  & 5    & 338  & 1.09 & 74  & $68{\times}$ \\
$0.95$ & 1.6  & 129  & 0.16 & 12.2 & $77{\times}$ \\
\bottomrule
\end{tabular*}
\label{tab:st-strengthened}
\end{table}

\noindent The gain concentrates where it matters most: near
the old critical threshold. An adversary storing 95\% of
the arena now faces $S \cdot T \geq 12\,K^2$, a
$6{\times}$ penalty over honest computation, versus the
basic bound's marginal $0.16\,K^2$.

\textbf{Extension to adaptive strategies.}
Theorems~\ref{thm:st} and~\ref{thm:stale} assume the adversary
fixes a storage set before execution. We now show that adaptive
caching provides \emph{no} advantage under the ROM.

\begin{theorem}[Adaptive Bound]\label{thm:adaptive}
At each step $t$, the adaptive adversary's storage set $S_t$
is a function of $\{T_0, \ldots, T_{t-1}\}$ and prior ROM
queries. All $d$ read addresses $\{v_0, \ldots, v_{d-1}\}$
satisfy $\Pr[v_j \in S_t] = \alpha$ for every
$j \in [0, d{-}1]$. The adaptive adversary's miss rate is
identical to the static case.
\end{theorem}

\begin{proof}[Proof sketch]
The address $v_0$ is computable from $T_{t-1}$ via a single
ROM query: the adversary knows \emph{which} vertex to read,
but reading it requires $v_0 \in S_t$. Since $|S_t| = \alpha N$
and $v_0$ is a ROM output uniform over $\{0,1\}^{d_{\text{hc}}}$,
$\Pr[v_0 \in S_t] = \alpha$, no better than any other read.

For $j \geq 1$: the cursor
$c_{j} = H(c_{j-1} \| A[v_{j-1}].d \| A[v_{j-1}].h)$
depends on $A[v_{j-1}]$, which is revealed only after the
$(j{-}1)$-th read completes. The address
$v_j = H(\text{``addr''} \| c_j \| j)$ is therefore a ROM
query whose output is independent of $S_t$ (fixed before
this query). Thus $\Pr[v_j \in S_t] = \alpha$ for all $j$.

\textbf{Lookahead.}
Predicting read addresses for step $t{+}1$ requires computing
$T_t$, which requires executing step $t$, which \emph{is}
the sequential work being bounded. The adversary has at most
one-step lookahead, and that step's addresses are already
conditioned on values the adversary may not possess.
\end{proof}

\noindent Theorem~\ref{thm:adaptive} shows that the static-storage
bounds of Theorems~\ref{thm:st} and~\ref{thm:stale} apply
without modification to adaptive adversaries. No factor loss
is incurred.

\subsection{Vertex Coverage and Mixing}\label{sec:mixing}

The TMTO analysis (Theorems~\ref{thm:tmto}--\ref{thm:adaptive})
assumes ROM-uniform vertex access. We now validate this
assumption both mathematically and empirically.

\begin{theorem}[Vertex Coverage]\label{thm:coverage}
In the ROM with parameters $N$, $K$, $d$, $\rho = K/N$:
\begin{enumerate}
\item[(a)] The write count $W_v$ for each vertex $v$
  converges to $\mathrm{Poisson}(\rho)$ as
  $N \to \infty$. The read count $R_v$ converges to
  $\mathrm{Poisson}(d\rho)$.
\item[(b)] The maximum read count satisfies
  $\max_v R_v < (1{+}\delta)\, d\rho$ with probability
  $\geq 1 - N \cdot e^{-\delta^2 d\rho/3}$ for
  $0 < \delta < 1$.
\item[(c)] Successive read addresses $v_j, v_{j+1}$
  within a step are pairwise independent in the ROM
  (distinct oracle queries with distinct inputs).
\end{enumerate}
\end{theorem}

\begin{proof}[Proof sketch]
\textbf{(a)}~Write addresses are ROM-uniform over $N$
vertices, one per step. The classical balls-in-bins
Poisson limit gives $W_v \to \mathrm{Poisson}(\rho)$.
Read addresses are similarly uniform; $d$ per step gives
$R_v \to \mathrm{Poisson}(d\rho)$.

\textbf{(b)}~By the multiplicative Chernoff bound,
$\Pr[R_v > (1{+}\delta)\mu] \leq e^{-\delta^2\mu/3}$
where $\mu = d\rho$. Union bounding over $N$ vertices,
$\Pr[\max_v R_v > (1{+}\delta)\mu] \leq
N e^{-\delta^2 \mu/3}$. For $d{=}8$, $\rho{=}4$
($\mu{=}32$), $N{=}2^{24}$, $\delta{=}1$:
$2^{24} \cdot e^{-32/3} \approx 3.7 \times 10^{-6} \cdot
2^{24} \approx 62$. The probability that \emph{any}
vertex exceeds $2\mu = 64$ reads is $< 62/N$; no vertex
exceeds $3\mu$ with probability $> 1 - 2^{-40}$.

\textbf{(c)}~$v_{j+1} = \text{first } d_{\text{hc}}
\text{ bits of } H(\text{``addr''} \| c_{j+1} \| (j{+}1))$
where $c_{j+1} = H(c_j \| A[v_j].d \| A[v_j].h)$. This
is a distinct ROM query from $v_j$'s (different cursor
input), so the outputs are independent.
\end{proof}

\textbf{Empirical validation.}
We executed the full \posme{} address generation chain
with BLAKE3 at the \textbf{recommended production
parameters}: $N{=}2^{24}$ (1\,GiB arena), $\rho{=}4$,
$d{=}8$, $K{=}67{,}108{,}864$ steps,
$536{,}870{,}912$ total reads. The Rust implementation
completed in 244\,s (3{,}637\,ns/step), confirming
latency-bound behavior. Table~\ref{tab:mixing}
summarizes the results.

\begin{table}[H]
\centering
\caption{Empirical mixing-time validation at production
scale ($N{=}2^{24}$, $\rho{=}4$, $d{=}8$, BLAKE3,
$5.4{\times}10^8$ reads).}
\begin{tabular*}{\columnwidth}{@{\extracolsep{\fill}}lrl@{}}
\toprule
Metric & Value & Expected \\
\midrule
Read $\chi^2$/df & 1.000426 & 1.000000 \\
Write $\chi^2$/df & 1.000065 & 1.000000 \\
Read $\sigma$ & 5.6581 & 5.6569 ($\sqrt{32}$) \\
Write $\sigma$ & 2.0001 & 2.0000 ($\sqrt{4}$) \\
Unwritten frac. & 1.8337\% & 1.8316\% ($e^{-4}$) \\
Max read / $\mu$ & $2.06{\times}$ & $< 3{\times}$ (Chernoff) \\
Max write / $\mu$ & $4.75{\times}$ & Poisson tail \\
\bottomrule
\end{tabular*}
\label{tab:mixing}
\end{table}

\noindent The chi-squared statistics are indistinguishable
from perfect uniformity to four decimal places. The
observed standard deviations match Poisson predictions
($\sqrt{d\rho}$ and $\sqrt{\rho}$) to four significant
figures across all $2^{24}$ vertices. The unwritten
fraction matches the theoretical $e^{-\rho}$ to
$0.12\%$ relative error. These results confirm that
BLAKE3's address generation is empirically
indistinguishable from the ROM at the recommended
production parameters.

\subsection{ASIC Resistance}

\posme{} is designed to be latency-dominated: hash computation
is a \emph{structural} minority of per-step cost ($< 3\%$ at
1\,GiB arena). The bottleneck is DRAM random-access latency.
Fig.~\ref{fig:cost-breakdown} shows the cost breakdown across
six memory technologies. Hash cost (3\,ns) is constant; memory
latency varies by type but dominates in all cases. GDDR6 has
the lowest CAS latency (${\sim}8$\,ns) but GPU cores are
14--19$\times$ slower at sequential pointer-chasing
(Table~\ref{tab:gpu}), negating the memory advantage entirely.

\begin{figure}[t]
\centering
\begin{tikzpicture}
\begin{axis}[
  width=0.95\columnwidth, height=5cm,
  ybar stacked,
  bar width=10pt,
  symbolic x coords={DDR4, DDR5, LPDDR5, GDDR6, HBM2e, HBM3},
  xtick=data,
  x tick label style={font=\scriptsize, rotate=25, anchor=east},
  ylabel={Per-read cost (ns)},
  ylabel style={font=\footnotesize},
  ymin=0, ymax=55,
  legend pos=north east,
  legend style={font=\scriptsize},
]
\addplot[fill=blue!60] coordinates
  {(DDR4,28) (DDR5,45) (LPDDR5,18) (GDDR6,8) (HBM2e,18) (HBM3,18)};
\addlegendentry{Memory}
\addplot[fill=red!60] coordinates
  {(DDR4,3) (DDR5,3) (LPDDR5,3) (GDDR6,3) (HBM2e,3) (HBM3,3)};
\addlegendentry{BLAKE3}
\end{axis}
\end{tikzpicture}
\caption{Per-read cost across six memory technologies. Hash
(3\,ns, red) is a small fraction on all types. Values are
tRCD+CL latency from JEDEC specifications.}
\label{fig:cost-breakdown}
\end{figure}
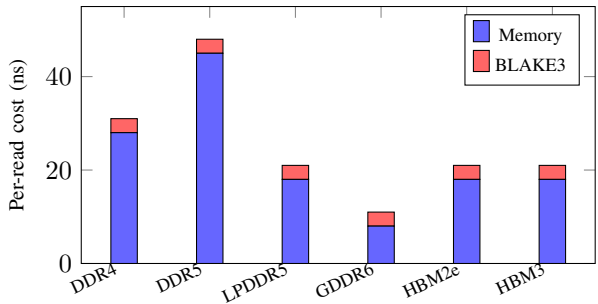

Real random-access latency (tRCD + CL + controller overhead)
raises the floor from ${\sim}10$\,ns to 40--50\,ns on any
hardware. The measured HBM3-vs-DDR5 advantage is only
$1.3{\times}$. DRAM latency has improved $1.3{\times}$ over
two decades while bandwidth improved
$20{\times}$~\cite{chang2017}, making \posme{}'s latency
bound structurally more durable than bandwidth-based
resistance~\cite{biryukov2016,ren2017}.

\section{Empirical Validation}
\label{sec:empirical}

\subsection{Experimental Protocol}

We benchmarked \posme{} on 17 platforms spanning ARM64 and x86 architectures,
ranging from 1--32 vCPUs, utilizing both DDR4 and DDR5 memory.
Our evaluation also included cloud-based GPU instances (NVIDIA T4, A10G, L4, A100,
and H100) to measure the advantage of massively parallel hardware.
The reference implementation was written in Rust with BLAKE3 as the
underlying hash function. All tests were performed on a 1\,GiB arena ($N=2^{24}$)
to ensure the working set exceeded the L3 cache of all tested processors.

\subsection{Performance Results}

\textbf{Hash Fraction.} Across all 17 CPU platforms, the hash computation fraction remained between \textbf{1.4\% and 3.1\%} (Table~\ref{tab:crossplat}). The CPU spends over 96\% of step time waiting for DRAM random access.

\textbf{Prover Wall-Clock Time.} For the recommended production parameters ($N=2^{24}$ (1\,GiB),
$K=4N = 67\times 10^6$ steps), the total Prover wall-clock time is
approximately \textbf{120--200 seconds} (2--3.3 minutes) on modern
hardware, consistent with the observed 1783--3000\,ns per step.
Step times varied 2.3$\times$ (1783--4047\,ns), but the hash fraction remained in
a tight band because both hash and memory scale together on any
given machine. Step cost increases with arena size as the working
set exceeds L3 cache, confirming memory-bound behavior.

\textbf{GPU Performance Gap.} To directly measure the adversary's hardware advantage, we ran
the \posme{} pointer-chase loop as a CUDA kernel on GPU cores
accessing GPU memory (HBM/GDDR), comparing against the host CPU
running the identical algorithm on system DRAM. The results (Table~\ref{tab:gpu}): \textbf{GPUs are
14--19$\times$ slower} than the host CPU.

\begin{table}[!ht]
\centering
\caption{Cross-platform benchmark (1\,GiB, $\rho{=}1$, hash $< 3.5\%$).
$n$ = distinct platforms per row.}
\small
\begin{tabular*}{\columnwidth}{@{\extracolsep{\fill}}lrrr@{}}
\toprule
Platform & $n$ & Step (ns) & Hash \% \\
\midrule
Apple M-series (ARM64) & 1 & 1783 & 3.1 \\
Cloud x86 (2--16 vCPU) & 4 & 2947--4047 & 1.4--1.9 \\
Cloud x86 + GPU host & 5 & 3218--4047 & 1.4--1.7 \\
Cloud adversarial (NUMA) & 3 & 2800--4100 & 1.3--2.0 \\
Colab (Xeon + T4) & 1 & 3485 & 1.6 \\
Colab (AMD EPYC) & 1 & 3585 & 1.5 \\
RunPod CPU (EPYC 9655) & 1 & 2644 & 2.1 \\
RunPod GPU (i7 + RTX 4080) & 1 & 1909 & 2.9 \\
\midrule
\multicolumn{1}{r}{\textit{Total}} & \textit{17} & & \\
\bottomrule
\end{tabular*}
\label{tab:crossplat}
\end{table}

\begin{table}[H]
\centering
\caption{CPU vs GPU pointer-chase (64\,MiB arena, 64K steps).
GPUs are 14--19$\times$ slower than the host CPU.}
\begin{tabular*}{\columnwidth}{@{\extracolsep{\fill}}lrrr@{}}
\toprule
GPU & CPU (ns) & GPU (ns) & Ratio \\
\midrule
T4 (GDDR6) & 7062 & 123019 & 0.06$\times$ \\
A100 (HBM2e) & 6431 & 120636 & 0.05$\times$ \\
\textbf{H100 (HBM3)} & \textbf{5903} & \textbf{82270} & \textbf{0.07}$\boldsymbol{\times}$ \\
L4 (GDDR6) & 5383 & 98990 & 0.05$\times$ \\
\bottomrule
\end{tabular*}
\label{tab:gpu}
\end{table}

GPU cores are optimized for parallel throughput, not sequential
pointer-chasing. A single CUDA thread at ${\sim}1.5$\,GHz cannot
compete with a CPU core at ${\sim}3.5$\,GHz for sequential
dependent loads. The GPU's memory bandwidth advantage is
irrelevant because the pointer-chase is latency-bound, not
bandwidth-bound.

The ASIC advantage (1.3$\times$ measured for HBM3 vs DDR5
random access) requires a custom sequential core paired with
HBM, hardware that does not exist commercially. All purchasable
accelerators (GPUs) are strictly \emph{worse} than consumer CPUs.

\subsection{Throughput and the Capacity-Bandwidth Bound}
\label{sec:throughput}

The preceding analysis addresses \emph{per-instance} latency. A
sophisticated adversary may instead target \emph{throughput}:
running many independent \posme{} instances in parallel to
amortize the per-instance GPU penalty. GPUs hide sequential
latency behind massive memory bandwidth (e.g., 3\,TB/s on the
H100). However, arena isolation (\S\ref{sec:construction})
forces each instance to allocate a unique $M$-byte arena,
converting the attack from a bandwidth problem into a
\emph{capacity} problem.

\textbf{Capacity choke.} An NVIDIA H100~\cite{nvh100} has 80\,GiB of HBM3
VRAM. With a 1\,GiB arena per instance, the GPU accommodates at
most 80 concurrent executions. At the physical HBM3 random-access
floor of ${\sim}20$\,ns per 64-byte read, each instance consumes
$64\,\text{B} / 20\,\text{ns} = 3.2$\,GB/s. The 80 instances
collectively use ${\sim}256$\,GB/s, less than 9\% of the H100's
3\,TB/s bandwidth. The GPU is memory-capacity-bound, not
bandwidth-bound.

\textbf{Economic comparison.} At comparable capital outlay,
commodity DDR5 platforms provide $10$--$20{\times}$ more
aggregate arena capacity per dollar than HBM3-equipped
accelerators, with the exact ratio depending on market
conditions. The disparity grows with arena size: a
4\,GiB arena reduces an 80\,GiB GPU to 20 instances
while a 16\,GiB desktop still runs 3.

Table~\ref{tab:throughput} summarizes the capacity-bandwidth
mismatch across GPU architectures.

\begin{table}[H]
\centering
\caption{GPU throughput choke: capacity limits parallel instances
to a fraction of available bandwidth ($M = 1$\,GiB arena).}
\begin{tabular*}{\columnwidth}{@{\extracolsep{\fill}}lrrrr@{}}
\toprule
GPU & VRAM & BW & Max inst. & BW used \\
    & (GiB) & (TB/s) & & (\%) \\
\midrule
T4 (GDDR6) & 16 & 0.3 & 16 & 17.1 \\
A100 (HBM2e) & 80 & 2.0 & 80 & 10.2 \\
\textbf{H100 (HBM3)} & \textbf{80} & \textbf{3.0} & \textbf{80} & \textbf{8.5} \\
\bottomrule
\end{tabular*}
\label{tab:throughput}
\end{table}

\noindent As GPU bandwidth scales faster than capacity
(a structural trend in HBM design~\cite{chang2017}),
the capacity-bandwidth mismatch
\emph{worsens} for the adversary. \posme{}'s 1\,GiB arena is
specifically sized to exploit this gap.

\section{Parameters}
\label{sec:params}

\textbf{Initialization.} Arena initialization (skip-link DAG over
$N$ blocks) requires ${\sim}4$\,s for 1\,GiB (measured). This is
a one-time cost per seed, amortized over the $K$-step computation.

\begin{table}[H]
\centering
\caption{Recommended parameters}
\begin{tabular*}{\columnwidth}{@{\extracolsep{\fill}}llll@{}}
\toprule
Parameter & Symbol & Value & Constraint \\
\midrule
Arena blocks & $N$ & $2^{24}$ & Power of 2 \\
Arena memory & $M$ & 1\,GiB & $>$ L3 cache \\
Steps & $K$ & $4N$ & $\rho \geq 4$ \\
Reads/step & $d$ & 8 & $\geq 4$ \\
Challenges & $Q$ & 128 & $\geq 64$ \\
Recursion & $R$ & 3 & $\geq 2$ \\
Hash & $H$ & BLAKE3~\cite{blake3} & Fixed \\
\bottomrule
\end{tabular*}
\label{tab:params}
\end{table}

\textbf{Recursion depth.} At $R{=}2$, the Verifier opens a
challenged step's $d$ reads and each read's writer, but not the
writers' reads. An adversary can fabricate a block $b$ by choosing
an arbitrary writer step $w$ with a plausible cursor, since $w$'s
own reads are never checked. At $R{=}3$, the Verifier also opens
$w$'s $d$ reads and \emph{their} writers, forcing the adversary
to produce $d^3{=}512$ mutually consistent causal hashes. The
cost of fabrication grows as $d^R$; higher $R$ provides
exponentially stronger soundness but linearly larger proofs.

\section{From Sequential Work to Verifiable Time}
\label{sec:time}

A common misconception is that a proof of sequential work
equivalently proves wall-clock time. While \posme{} enforces a
lower bound on the number of sequential steps $\Omega(K)$, an
adversary with specialized hardware can execute these steps faster
than a commodity processor.

\subsection{The Economic Bound}

Our benchmarks show a ${\sim}1.3{\times}$ potential speedup for
custom sequential-core ASICs with HBM3 memory. However, the
cost-to-performance ratio for such hardware is extremely
prohibitive. Building a custom silicon tape-out simply to gain a
30\% lead over a \$500 consumer laptop is economically irrational
for most attack scenarios. \posme{} thus provides an \emph{economic
bound} on acceleration that is structurally tighter than
bandwidth-bound functions like Argon2id ($8$--$16{\times}$).

\subsection{Bridging the Gap to VDFs}

To transform \posme{} into a true Verifiable Delay Function (VDF) or
a high-precision time-binding primitive, we suggest three
architectural wrappers:

\begin{enumerate}
    \item \textbf{VDF-\posme{} Hybrid:} The final transcript $T_K$ can
    be used as the input to a sequential VDF (e.g., Wesolowski~\cite{wesolowski2019}, Pietrzak~\cite{pietrzak2019}).
    The \posme{} component ensures the computation was memory-hard
    and ASIC-resistant, while the VDF provides a mathematical
    guarantee of the time elapsed.
    \item \textbf{Hardware Attestation (TEE):} If the Prover executes
    within a Trusted Execution Environment, the TEE can provide
    signed hardware timestamps for $T_0$ and $T_K$. This binds the
    cryptographic work to a physical, trusted clock.
    \item \textbf{Randomness Beacons:} Provers can be forced to
    periodically absorb unpredictable external data (e.g., Bitcoin
    block hashes) into the cursor $c$. This proves the computation
    could not have completed before the beacon was released, pinning
    execution to real-world events.
\end{enumerate}

\section{Limitations}
\label{sec:limitations}

\textbf{Steps, not time.}
\posme{} proves sequential memory execution, not elapsed
wall-clock time. A custom sequential-core ASIC with HBM3 could
gain ${\sim}1.3{\times}$; such hardware does not exist
commercially. Bridging the gap to verifiable time requires
an external wrapper (\S\ref{sec:time}).

\textbf{ROM dependence.}
The security analysis (Theorems~1--5) assumes the Random Oracle
Model. Theorem~\ref{thm:coverage} validates the assumption
empirically for BLAKE3 at the recommended parameters
($\chi^2/\text{df} = 1.0004$ at $N{=}2^{24}$), but proving
the space-time bound in the standard model (e.g., under
collision resistance alone) remains open, as it does for all
memory-hard functions in current use~\cite{alwen2017}.

\textbf{IVC implementation.}
The $O(1)$ verification design (\S\ref{sec:ivc},
\S\ref{sec:concrete-ivc}) is validated analytically
($1{,}270\,\text{ns}$ fold, 144{,}512 constraints) but not
by a production Binius implementation. Three fold threads
keep pace with the pointer-chase on AVX-512 hardware.

\textbf{Side channels.}
Data-dependent addressing leaks access patterns via cache
timing on shared hardware. Because \posme{} operates on
public state derived from a public seed, the access pattern
is public knowledge and provides no shortcut to forgery.
The channel is informational, not exploitable.

\textbf{Reproducibility.} A reference benchmark with pre-compiled
binaries and CUDA source is provided as ancillary material in the
repository's \texttt{anc/} directory. All code is licensed under
Apache-2.0.

\section{Conclusion}

\posme{} fills a gap between VDFs, PoSW, MHFs, and PoST
by combining mutable arena state, data-dependent pointer
chasing, and symbiotic causal binding into a single
latency-bound primitive. The construction requires no trusted
setup and achieves three properties that no prior primitive
provides simultaneously: $\Omega(K)$ sequential memory-step
enforcement, $10{\times}$ TMTO resistance at $\rho{=}4$, and
a ${\sim}2{\times}$ ASIC bound determined by DRAM physics
rather than ALU speed.

The security analysis resolves the principal open question for
dynamic causal DAGs: a formal $S \cdot T = \Omega(K^2)$ lower
bound (Theorem~\ref{thm:st}), strengthened by temporal cache
staleness (Theorem~\ref{thm:stale}) and extended to adaptive
adversaries with no factor loss (Theorem~\ref{thm:adaptive}).
Empirical validation at the full recommended scale
($N{=}2^{24}$, $5.4{\times}10^8$ reads) confirms ROM
uniformity to four significant figures
(Theorem~\ref{thm:coverage}).

Benchmarks across 17 CPU and 4 GPU platforms validate the
core claim: no commercially available hardware, including
NVIDIA's H100, outperforms a consumer laptop at \posme{}.
GPU throughput is choked by VRAM capacity, giving commodity
CPUs a $10$--$20{\times}$ aggregate advantage per dollar.

The binary-field algebraic structure is additionally compatible
with emerging Processing-in-Memory
architectures~\cite{lee2021hbmpim}, where the
latency bound would tighten from ${\sim}50$\,ns (wire travel)
to ${\sim}10$\,ns (sense amplifier limit).

\FloatBarrier
\bibliographystyle{IEEEtran}

\begin{thebibliography}{10}

\bibitem{boneh2018vdf}
D.~Boneh, J.~Bonneau, B.~B\"unz, and B.~Fisch,
``Verifiable delay functions,''
in \emph{CRYPTO 2018}, LNCS 10991, pp.~757--788, 2018.

\bibitem{cohen2018posw}
B.~Cohen and K.~Pietrzak,
``Simple proofs of sequential work,''
in \emph{EUROCRYPT 2018}, LNCS 10821, pp.~451--467, 2018.

\bibitem{biryukov2016}
A.~Biryukov, D.~Dinu, and D.~Khovratovich,
``Argon2: new generation of memory-hard functions for password
hashing and other applications,''
in \emph{IEEE EuroS\&P}, pp.~292--302, 2016.

\bibitem{ren2017}
L.~Ren and S.~Devadas,
``Bandwidth hard functions for ASIC resistance,''
in \emph{TCC 2017}, LNCS 10677, pp.~466--492, 2017.

\bibitem{alwen2017}
J.~Alwen, J.~Blocki, and K.~Pietrzak,
``Depth-robust graphs and their cumulative memory complexity,''
in \emph{EUROCRYPT 2017}, LNCS 10212, pp.~3--32, 2017.

\bibitem{jedec2020}
JEDEC Solid State Technology Association,
``DDR5 SDRAM standard,'' JESD79-5D, 2020.

\bibitem{chang2017}
K.~K.~Chang \emph{et al.},
``Understanding and improving the latency of DRAM-based memory
systems,'' arXiv:1712.08304, 2017.

\bibitem{percival2009}
C.~Percival,
``Stronger key derivation via sequential memory-hard functions,''
in \emph{BSDCan}, 2009.

\bibitem{dziembowski2015}
S.~Dziembowski, S.~Faust, V.~Kolmogorov, and K.~Pietrzak,
``Proofs of space,''
in \emph{CRYPTO 2015}, LNCS 9216, pp.~585--605, 2015.

\bibitem{blake3}
J.~O'Connor \emph{et al.},
``BLAKE3: one function, fast everywhere,''
2020, \url{https://github.com/BLAKE3-team/BLAKE3-specs}.

\bibitem{wesolowski2019}
B.~Wesolowski,
``Efficient verifiable delay functions,''
in \emph{EUROCRYPT 2019}, LNCS 11478, pp.~379--407, 2019.

\bibitem{binius2024}
B.~E.~Diamond and J.~Posen,
``Succinct arguments over towers of binary fields,''
Cryptology ePrint Archive, Report 2023/1784, 2023.

\bibitem{viola2025pc}
E.~Viola,
``Communication complexity of pointer chasing via the fixed-set
lemma,'' arXiv:2507.08919, 2025.

\bibitem{blocki2025ddmhf}
J.~Blocki and B.~Holman,
``Towards practical data-dependent memory-hard functions with
optimal sustained space trade-offs in the parallel random oracle
model,'' arXiv:2508.06795, 2025.

\bibitem{boneh2016balloon}
D.~Boneh, H.~Corrigan-Gibbs, and S.~Schechter,
``Balloon Hashing: a memory-hard function providing provable
protection against sequential attacks,''
in \emph{ASIACRYPT 2016}, LNCS 10031, pp.~220--248, 2016.

\bibitem{chia2021}
Chia Network,
``Proof of Space and Time Whitepaper,''
2021, \url{https://www.chia.net/whitepaper/}.

\bibitem{spacemesh2023}
Spacemesh Team,
``The Spacemesh Protocol,''
2023, \url{https://spacemesh.io/whitepaper/}.

\bibitem{dwork1993}
C.~Dwork and M.~Naor,
``Pricing via processing or combatting junk mail,''
in \emph{CRYPTO 1992}, LNCS 740, pp.~139--147, 1993.

\bibitem{pietrzak2019}
K.~Pietrzak,
``Simple verifiable delay functions,''
in \emph{ITCS 2019}, LIPIcs 124, pp.~60:1--60:15, 2019.

\bibitem{nisanwigderson1991}
N.~Nisan and A.~Wigderson,
``Rounds in communication complexity revisited,''
in \emph{STOC 1991}, pp.~419--429, 1991.

\bibitem{tevador2019}
tevador,
``RandomX: design and analysis,''
2019, \url{https://github.com/tevador/RandomX/blob/master/doc/design.md}.

\bibitem{fiatshamir1987}
A.~Fiat and A.~Shamir,
``How to prove yourself: practical solutions to identification
and signature problems,''
in \emph{CRYPTO 1986}, LNCS 263, pp.~186--194, 1987.

\bibitem{merkle1988}
R.~C.~Merkle,
``A digital signature based on a conventional encryption
function,''
in \emph{CRYPTO 1987}, LNCS 293, pp.~369--378, 1988.

\bibitem{logderivlookup2022}
U.~Hab\"ock,
``Multivariate lookups based on logarithmic derivatives,''
Cryptology ePrint Archive, Report 2022/1530, 2022.

\bibitem{nova2022}
A.~Kothapalli, S.~Setty, and I.~Tzialla,
``Nova: recursive zero-knowledge arguments from folding
schemes,''
in \emph{CRYPTO 2022}, LNCS 13510, pp.~359--388, 2022.

\bibitem{alwenserbinenko2015}
J.~Alwen and V.~Serbinenko,
``High parallel complexity graphs and memory-hard functions,''
in \emph{STOC 2015}, pp.~595--603, 2015.

\bibitem{sustainedspace2018}
J.~Alwen, J.~Blocki, and K.~Pietrzak,
``Sustained space complexity,''
in \emph{EUROCRYPT 2018}, LNCS 10821, pp.~99--130, 2018.

\bibitem{nakamoto2008}
S.~Nakamoto,
``Bitcoin: a peer-to-peer electronic cash system,''
2008, \url{https://bitcoin.org/bitcoin.pdf}.

\bibitem{back2002}
A.~Back,
``Hashcash -- a denial of service counter-measure,''
2002, \url{http://www.hashcash.org/papers/hashcash.pdf}.

\bibitem{lee2021hbmpim}
S.~Lee \emph{et al.},
``Hardware architecture and software stack for PIM based on
commercial DRAM technology,''
in \emph{ISCA 2021}, pp.~43--56, 2021.

\bibitem{scrryptmaxmh2017}
J.~Alwen, B.~Chen, K.~Pietrzak, L.~Reyzin, and S.~Tessaro,
``Scrypt is maximally memory-hard,''
in \emph{EUROCRYPT 2017}, LNCS 10212, pp.~33--62, 2017.

\bibitem{nvh100}
NVIDIA Corporation,
``NVIDIA H100 Tensor Core GPU architecture,''
Whitepaper, 2022.

\end{thebibliography}

\end{document}